\documentclass{sig-alternate-10pt}

\usepackage{cite}

\usepackage{wrapfig}

\usepackage{graphicx}
\usepackage{xcolor}

\usepackage{fancybox}
\usepackage{calc}
\usepackage{url}
\usepackage{flushend} 


\newcommand{\cut}[1]{}

\definecolor{termcolor}{RGB}{102,138,201}
\newcommand{\term}[1]{\textcolor{termcolor}{\bf #1}}

\newtheorem{theorem}{Theorem}
\newdef{definition}{Definition}

\makeatletter
\let\@copyrightspace\relax
\makeatother

\newcommand{\captionfonts}{\small}

\makeatletter
\long\def\@makecaption#1#2{\baselineskip 10pt
   \vskip 10pt
   \setbox\@tempboxa\hbox{\textbf{#1: #2}}
   \ifdim \wd\@tempboxa >\hsize 
       \textbf{\captionfonts #1: #2}\par                
     \else                      
       \hbox to\hsize{\hfil\box\@tempboxa\hfil}
   \fi}
\makeatother

\begin{document}


\title{BGP Stability is Precarious}

\numberofauthors{1}

\author{\alignauthor P. Brighten Godfrey\\
\affaddr{University of Illinois at Urbana-Champaign}\\
\email{pbg@illinois.edu}}


\maketitle


\begin{abstract} \textnormal{We note a fact which is simple, but may be useful for the networking research community: essentially {\em any} change to BGP's decision process can cause divergence --- or convergence when BGP would otherwise diverge.} \end{abstract}


\section{Introduction}
\label{sec:intro}

The Internet's interdomain routing protocol, BGP~\cite{rfc1771}, uses a decision process to select a single best route to each destination when presented with multiple options.  This decision process can be customized and modified at each router to select routes that achieve various objectives such as load balance, path quality, or security.  When proposing such a modification, we were asked a very natural question: Given the known problem that a distributed network of BGP routers might never converge to a stable state~\cite{varadhan2000persistent}, might the proposed change make the problem worse?  That is, do there exist cases in which the standard BGP protocol converges, but the proposed modification causes divergence?

The Hippocratic goal to do no harm is natural to desire of any modification to BGP's decision process, given its global importance.  However, we observe here that {\em any} modification to the decision process can cause divergence in some case when standard BGP would converge, under very mild conditions.  Specifically, (1) the modified BGP must actually differ, in that there is some case where the modified and standard BGP both converge, but to different outcomes; and (2) the modification either may be deployed at only some routers, or the modification preserves the expressiveness of standard BGP (for example by maintaining the initial operator-configurable LOCAL\_PREF step).  Even seemingly trivial changes, like changing a tiebreaking step from ``lowest router ID'' to ``highest router ID'', satisfy these conditions and therefore may cause divergence.

But this fact should not incite fear of modifying BGP.  Indeed, any modification could also cause {\em convergence} when BGP would otherwise diverge. Thus, fear of modifying BGP can be equally matched with fear of {\em not} modifying BGP.

Instead, what the result points out is that the question of whether new cases of divergence \emph{could} happen by switching from decision process $A$ to $B$ is uninformative, because the answer is always ``yes'' for any distinct values of $A$ and $B$.  A more valuable question is how convergence is affected in realistic cases.  This, of course, is a much more difficult question to answer convincingly, not least because it requires assumptions about what is realistic.

\section{Model}
\label{sec:model}

\subsection{The standard model}

We follow the model of~\cite{griffin02stable}.\footnote{We omit \cite{griffin02stable}'s FIFO queues and permitted path sets, which other features of the model can emulate.}  An instance of the Stable Paths Problem (SPP) consists of a graph $G=(V,E)$ and a set $\lambda$ of \term{ranking functions}, one for each node $v\in V$.  Node $v$'s ranking function $\lambda_v$ specifies which paths $v$ prefers; specifically, if $\lambda_v(P_1) > \lambda_v(P_2)$ then $v$ prefers $P_1$ over $P_2$. We require that $\lambda_v(P_1) \neq \lambda_v(P_2)$ unless $P_1$ and $P_2$ have the same first edge (since BGP learns only a single route from each neighbor, we will never need to compare two such routes).

The ``null path'' $\varepsilon$ represents the absence of a path to the destination, and is considered a valid path.  Since BGP's decision process may eliminate a path $P$ due to import or export filters, we may have $\lambda_v(\varepsilon) > \lambda_v(P)$.  We write $P_1 P_2$ to denote the concatenation of two paths, or $vwP$ to concatenate the edge $(v,w)$ with path $P$.  


There is a single distinguished node $0$ to which all nodes are choosing paths.  At any given time $t$, each node $v$ has a current \term{path assignment} $\pi_t(v)$.  At all times $t$, we have $\pi_t(0) = 0$ (i.e., the destination always selects the trivial one-hop path to itself).  The dynamics of the protocol are modeled by a sequence of ``activations'' of nodes $A= (v_{i_1}, v_{i_2}, \ldots)$ in which each node (other than $0$) must appear infinitely often.  At time $t$, only node $A_t$ updates its selected route $\pi_t(A_t)$ and all other nodes are unaffected.  Specifically, if $v = A_t$, then $v$ chooses its new best route by setting \begin{equation} \pi_t(v) = \textrm{argmax}_{P \in choices(v,t)} \lambda_v(P), \label{eq:update}\end{equation} where $choices(v,t)$ is the set of all simple (non-loopy) paths of the form $vw\pi_t(w)$ where $w$ is a neighbor of $v$ and $\pi_t(w)$ is $w$'s current path.


A node $v$ is \term{stable} in path assignment $\pi$ if executing (\ref{eq:update}) produces no change.  A path assignment $\pi$ is stable if all nodes are stable in $\pi$.   An instance $(G,\lambda)$ is \term{safe} if any activation sequence eventually produces a stable path assignment, regardless of the initial path assignment.

\subsection{Modeling a modified decision process}

A ranking function $\lambda$ encapsulates the final result of the BGP decision process, whether that is due to an operator's assignment of the LOCAL\_PREF attribute for a route, or minimizing the AS\_PATH length, or any of the various other factors that affect the decision process.  So a ``modified BGP decision process'' is simply a different ranking function $\lambda'$.

But two ranking functions might be effectively equivalent, in that they produce the same outcome in practice.  The following definition rules out such degenerate modifications.

\begin{definition} Two ranking functions $\lambda,\lambda'$ are \term{safely distinct} if there exists a network $N$ for which $(N,\lambda)$ and $(N,\lambda')$ are safe, but their stable states differ.
\end{definition}

(Note that since the two instances are safe, they each have a single stable state~\cite{sami2009searching,jaggard11distributed}.)  As mentioned in the introduction, this definition restricts our attention to the case that there is some network on which $\lambda$ and $\lambda'$ are safe and converge to different outcomes.  While this appears to be a very mild restriction, it is conceivable that $\lambda$ and $\lambda'$ \emph{always} produce identical stable states \emph{except} on networks where at least one of them may diverge.  In that case, reasoning about differences in outcomes involves the system's dynamics, i.e., particular activation sequences.  It would be possible to use our technique to make statements about particular activation sequences, but we choose to avoid that complication here.

\section{Precariousness}
\label{sec:theorem}

\subsection{Partial deployment}

In this section we show that any safely distinct modification of the BGP decision process can cause divergence or convergence, when partially deployed.

But what exactly is a ``partial deployment'' of the modified decision process? Since a ranking function is defined for a specific network, how can we ``deploy'' it in a new environment where it may have to rank new paths?  Fortunately we can sidestep this modeling complication since we will need to use the ranking functions in only a black-box manner in our theorem.

Specifically, suppose we have ranking functions $\lambda^N$ and $\lambda^G$ on networks $N$ and $G$, respectively, and a given subgraph $N' \subseteq G$ is identical to $N$. Then a \term{partial deployment} of $\lambda^N$ in $(G,\lambda^G)$ is an instance $(G,\lambda^*)$ where
\[
\lambda_{v}^{*}(P)=\left\{ \begin{array}{ll}
\lambda_{v}^{G}(P) & \mbox{if }v\in G\setminus N'\\
\lambda_{v}^N(P) & \mbox{if }v\in N'\mbox{ and }P\subseteq N'\\
-\infty & \mbox{if }v\in N'\mbox{ and }P\not\subseteq N'.\end{array}\right.\]
In other words, the new ranking function $\lambda^*$ mimics $\lambda^G$ except on $N'$ where it mimics $\lambda^N$.  The third case causes $\lambda^*$ to rank any path outside $N'$ strictly less than $\varepsilon$, which ensures that $\lambda^N$ never is called upon to rank a path outside the network $N'$ on which it is well-defined.  This models a scenario in which nodes outside $N'$ export no BGP route advertisements to nodes in $N'$.

We can now state and prove the theorem.

\begin{theorem} \label{thm:partial} If $\lambda$ and $\lambda'$ are safely distinct, then there exists an SPP instance $(G,\lambda^G)$ in which a partial deployment of $\lambda$ is safe, but a partial deployment of $\lambda'$ has no stable path assignment.
\end{theorem}

One can interpret the theorem as follows.  If we let $\lambda$ be the behavior of standard BGP, then the partial deployment of $\lambda$ in $(G,\lambda^G)$ just means that the whole network runs standard BGP, and the modification $\lambda'$ causes divergence.  Symmetrically, we can just as easily let $\lambda'$ be the behavior of standard BGP, in which case the modification causes convergence.

\begin{proof} We construct $G$ as follows (Fig. \ref{fig:gadget1}). Since $\lambda$ and $\lambda'$ are safely distinct, there is a network $N$ on which their stable states differ.  We include in $G$ two copies of $N$ which we call $N$ and $N'$, but with only one instance of the destination $0$.  By the condition of the theorem, there must exist a $w \in N$ which has differing path selections in the stable states of $(N,\lambda)$ and $(N,\lambda')$.  Let $w'$ be the corresponding node in $N'$. We add a new node $x$ connected to $w$ and $w'$.  Finally, we add an ``oscillator gadget'' --- a triangle $a,b,c$ with each node connected to the destination $0$ --- and connect $a$ to $x$.

We construct $\lambda^G$ as follows. First, $\lambda^G_v = \lambda_v$ for all $v \in N$.  The behavior of $\lambda^G$ on $N'$ is irrelevant, since this is where we will place the partial deployment of $\lambda$ or $\lambda'$.

Second, $\lambda^G_x$ ranks paths as follows.  Let $P_1, \ldots, P_k$ be a list of all $w \leadsto 0$ paths in $N$, and let $P'_1, \ldots, P'_k$ be the corresponding $w' \leadsto 0$ paths in $N'$.  Without loss of generality, suppose that $P_1$ is $w$'s selected path in the stable state of $(N,\lambda')$, while $w$'s selected path in $(N,\lambda)$ is some other path $P_i$.  Then we let $\lambda^G_x(xwP_1) > \lambda^G_x(xw'P'_1) > \lambda^G_x(xwP_2) > \lambda^G_x(xw'P'_2) > \ldots > \lambda^G_x(xwP_k) > \lambda^G_x(xw'P'_k) \\ > \varepsilon,$ with all other paths ranked below $\varepsilon$.

Third and finally, on the oscillator gadget, $\lambda_G$ behaves like the classic Bad Gadget~\cite{griffin02stable}: each of $a,b,c$ will accept one of two paths, the direct path (e.g. $a0$) and the path via its counterclockwise neighbor (e.g. $ab0$), with the latter preferred.  However, to this structure we add the fact that $a$ most prefers the path $axwP_i$.

	\begin{figure}[t]
	\centering
	\includegraphics[width=2.6in]{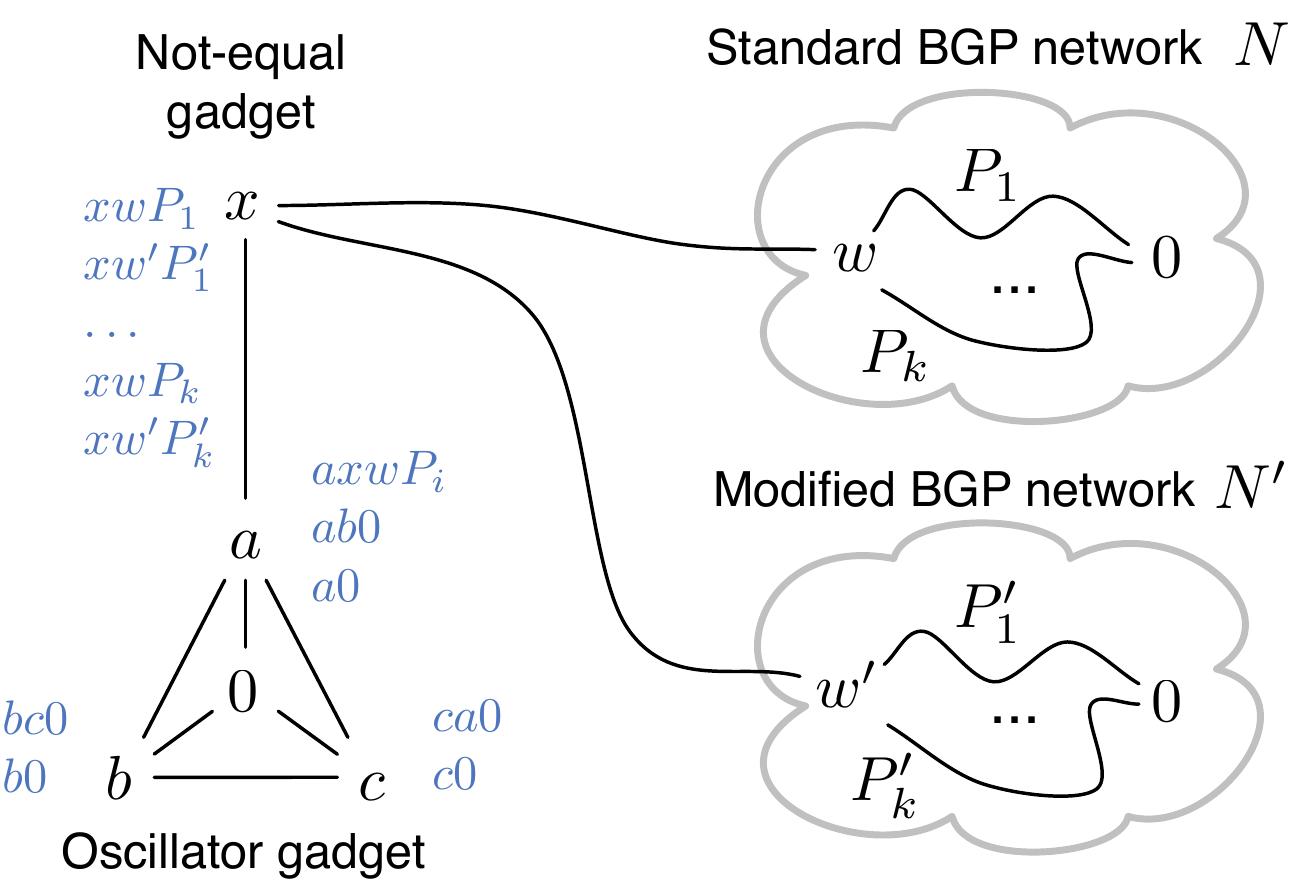}
	\caption{An SPP instance which converges if and only if $N$ and $N'$ have the same stable state.  The ranking function of certain nodes is written in blue next to the node, listing paths from most to least preferred. Multiple copies of the destination $0$ are drawn for clarity, but these are in fact the same node.}
	\label{fig:gadget1}
	\hrulefill
	\end{figure}

With a partial deployment of $\lambda$ on $N'$, the effect of the construction is as follows. Since $(N,\lambda)$ is safe, it must have a single stable state~\cite{sami2009searching,jaggard11distributed}, so $N$ and $N'$ will eventually stabilize with corresponding path selections $P_i$ and $P_i'$.  Therefore, since $x$ always prefers a path in $N$ over the corresponding path in $N'$, it will eventually select $xwP_i$ permanently, causing $a$ to select the path $axwP_i$, causing $c$ to select $c0$, and $b$ to select $bc0$.  Thus, a unique stable state is reached for any activation sequence.

On the other hand, consider a partial deployment of $\lambda'$ on $N'$. Since $(N',\lambda')$ is safe it must have a single stable state~\cite{jaggard11distributed}, which we know must differ from the stable state of $(N,\lambda)$. Thus, after $N$ and $N'$ converge, node $x$ is presented with two {\em different} paths, $xwP_i$ and $xw'P'_1$.  It will prefer the path $xw'P'_1$ and remain with that selection thereafter.  With $a$'s possibility of any more-preferred path via $x$ now eliminated, the nodes $a,b,c$ mimic the Bad Gadget, and have no stable state.  Therefore, with a partial deployment of $\lambda'$, this instance has no stable state. \hfill{}
\end{proof}

\subsection{Full deployment}

If the requirement of partial deployment were removed, Theorem~\ref{thm:partial} would no longer hold.  Consider, for example, a modified decision process which simply performs shortest path routing.  The theorem shows that a partial deployment of shortest path routing can cause divergence; however, a full deployment will always converge.

But the theorem holds with full deployments if we add a constraint on the modification: it must preserve the expressive power of BGP.  One way to formalize this is as follows.  The operator of each node $v$ specifies a \term{partial ranking function} $\hat{\lambda}_v$ which may assign multiple paths the same value.  A \term{decision process} is now a function $d$ which, given a partial ranking function $\hat{\lambda}_v$, returns a ranking function $d_{\hat{\lambda}_v}$  consistent with $\hat{\lambda}_v$ (that is, $\hat{\lambda}_v(P_1) > \hat{\lambda}_v(P_2)$ implies $d_{\hat{\lambda}_v}(P_1) > d_{\hat{\lambda}_v}(P_2)$).  Intuitively, $d$ breaks any ``ties'' in $\hat{\lambda}$'s ranking of paths.\footnote{Recall from the definition of a ranking function that $d_{\hat{\lambda}_v}$ might still have ties, but only between two routes that go through the same neighbor, which we never need to compare.}  We let $d_{\hat{\lambda}}$ refer to the set of ranking functions produced by applying $d$ to the set of partial ranking functions $\hat{\lambda}_v$ over all nodes $v$.

Taking common ISP business relationships as an example, an operator might specify $\hat{\lambda}_v(P) = 100$ for paths through $v$'s providers, $\hat{\lambda}_v(P) = 200$ for paths through peers, and $\hat{\lambda}_v(P) = 300$ for paths through customers.  If $v$ has multiple providers, peers, or customers, this will not always yield a unique best path; the decision process $d$ breaks those ties, perhaps by examining path length or other factors. However, the operator {\em can} choose to specify an arbitrary ranking function by giving $\hat{\lambda}_v(P)$ a distinct value for each $P$ (in which case $d$ does not affect the outcome). In this sense, any modified decision process preserves BGP's expressiveness.

We can now show a theorem analogous to Theorem \ref{thm:partial} without the partial deployment requirement.  The proof is an easy adaptation of our earlier technique: since the partial ranking function can be used to force any total order, we can  build the necessary ranking functions even though the modification is deployed at all nodes.

\begin{definition} Two decision processes $d,d'$ are \term{safely distinct} if there exists a network $N$ and partial ranking functions $\hat{\lambda}$ for which $(N,d_{\hat{\lambda}})$ and $(N,d'_{\hat{\lambda}})$ are safe, but their stable states differ.
\end{definition}

\begin{theorem} \label{thm:full} If $d$ and $d'$ are safely distinct, then there exists a network $G$ and partial ranking functions $\hat{\lambda}^G$ in which $(G,d_{\hat{\lambda}^G})$ is safe, but $(G,d'_{\hat{\lambda}^G})$ has no stable path assignment.
\end{theorem}

\begin{proof} Let $N$ and $\hat{\lambda}$ be such that $(N,d_{\hat{\lambda}})$ and $(N,d'_{\hat{\lambda}})$ are safe, but their stable states differ.  We construct $G$ based on $N$ as in the proof of Theorem~\ref{thm:partial}.  Let $\lambda^G_{v}$ be the ranking functions constructed in the proof of Theorem~\ref{thm:partial} with $\lambda = d_{\hat{\lambda}}$ and $\lambda' = d'_{\hat{\lambda}}$. Define the partial ranking functions $\hat{\lambda}^G$ as  
\[
\hat{\lambda}^G_{v}=\left\{ \begin{array}{ll}
\lambda^G_{v} & \mbox{if }v\in \{x,a,b,c\}\\
d_{\hat{\lambda}_v} & \mbox{if }v\in N\\
\hat{\lambda}_{v} & \mbox{if }v\in N'.\end{array}\right.\]
Note that applying one of the decision processes to $\hat{\lambda}^G$ can only vary its behavior for $v\in N'$. With this construction, applying the decision process $d$ yields ranking functions $d_{\hat{\lambda}^G}$ which are exactly equivalent to $\lambda^G$ with a partial deployment of $d_{\hat{\lambda}}$ on $N'$.  Likewise, $d'_{\hat{\lambda}^G}$ is exactly equivalent to $\lambda^G$ with a partial deployment of $d'_{\hat{\lambda}}$ on $N'$.  Therefore, the result follows by the argument of the proof of Theorem~\ref{thm:partial}. \hfill{} \end{proof}

\section{Extensions}

Our results dealt with decision processes which differ in their final stable state.  One can also show that if there is any difference in path selections in two ranking functions $\lambda$ and $\lambda'$ {\em at any moment} during the dynamic convergence process, then for a particular activation sequence, a partial deployment of $\lambda'$ will diverge while a partial deployment of $\lambda$ will converge (or vice versa).  This involves inserting a gadget (essentially Disagree~\cite{griffin02stable}) between $x$ and $a$ in the construction of Fig.~\ref{fig:gadget1}, to ``remember'' that some difference has occurred in the past.  In one sense, this is stronger than our previous results, as it applies even to transient differences between $\lambda$ and $\lambda'$.  However, it is less satisfying since it needs an activation sequence that runs $N$ and $N'$ in lockstep.  With any variation in timing, even two copies of $\lambda$ could be judged to be different at some moments in time.

One could consider more general models of the BGP decision process, perhaps treating it as a state machine with memory, in order to model features such as route flap damping~\cite{rfc2439} which change their preferences across time. Since our theorems use $\lambda$ and $\lambda'$ essentially as black boxes, such extensions may be straightforward.

\section{Acknowledgements}

We thank Alex Fabrikant, Michael Schapira, and Scott Shenker for helpful comments.

\bibliographystyle{plain}
\bibliography{paper}

\end{document}